\documentclass[12pt]{article}
    \usepackage[margin=1in]{geometry}
    \usepackage[utf8]{inputenc}
    \usepackage[T1]{fontenc}    
    \usepackage{amsmath,amssymb,amsfonts,amsthm}
    \usepackage{mathtools}
    \usepackage{physics}
    \usepackage[symbol]{footmisc}
    \usepackage{url}
    \usepackage{bigints}
    \usepackage{graphicx}
    \usepackage{tabularx}
    \usepackage{xcolor}
    \usepackage{todonotes}
\newtheorem{theorem}{Theorem}
\theoremstyle{definition}

\theoremstyle{remark}

\graphicspath{{./images/}}

\begin{document}

\begin{titlepage}
\begin{center}
    \vspace{0.5cm}

    \LARGE\textbf{Energy translation symmetries and dynamics of separable autonomous two-dimensional ODEs}
       
    \normalsize
        
    \vspace{1.0cm}
    \setcounter{footnote}{0}
    \setlength{\footnotemargin}{0.8em}
    {\large Johannes G. Borgqvist\footnote{Corresponding author. E-mail: borgqvist@maths.ox.ac.uk}\footnote{\label{Oxford}Wolfson Centre for Mathematical Biology, Mathematical Institute, University of Oxford, United Kingdom}, Fredrik Ohlsson\footnote{\label{Umeå}Department of Mathematics and Mathematical Statistics, Umeå University, Sweden}, and Ruth E. Baker\footref{Oxford}}
    \setlength{\footnotemargin}{1.8em}
        
    \vspace{1cm}
        
    \begin{abstract}
    We study symmetries in the phase plane for separable, autonomous two-state systems of ordinary differential equations (ODEs). We prove two main theoretical results concerning the existence and non-triviality of two orthogonal symmetries for such systems. In particular, we show that these symmetries correspond to translations in the internal energy of the system, and describe their action on solution trajectories in the phase plane. In addition, we apply recent results establishing how phase plane symmetries can be extended to incorporate temporal dynamics to these energy translation symmetries. Subsequently, we apply our theoretical results to the analysis of three models from the field of mathematical biology: a canonical biological oscillator model, the Lotka--Volterra (LV) model describing predator-prey dynamics, and the SIR model describing the spread of a disease in a population. We describe the energy translation symmetries in detail, including their action on biological observables of the models, derive analytic expressions for the extensions to the time domain, and discuss their action on solution trajectories. 
    \end{abstract}
    
    \vspace{0.5cm}
    
\end{center}
\end{titlepage}
   
\setcounter{footnote}{0}
\renewcommand{\thefootnote}{\arabic{footnote}}

\numberwithin{equation}{section}

\section{Introduction}
In the most straightforward case, minimal biological models are analysed by means of a linear stability analysis in the phase plane~\cite{murray2002}. Typically, this approach entails analysing a time-invariant system of two first order ODEs in the phase plane, which is the plane spanned by the two states of the model, to provide qualitative information about the long-term dynamics of the system. However, such analysis does not provide quantitative insight into the relationship between different solution trajectories of the same model and it cannot identify common structural properties of the solution trajectories. Even in the simplest case of separable models in the phase plane which are directly solvable, albeit usually implicitly, these questions about the properties of the model, including the relationship between different solution trajectories, cannot be answered by a linear stability analysis. Symmetry methods, however, provide a versatile and generalisable set of mathematical tools for answering these types of questions. They have been used with huge success in theoretical physics but are not yet widely used in mathematical biology. 

For coupled systems of two or more first order ODEs there are few systematic methods for finding Lie symmetries. For this type of systems, the \textit{linearised symmetry conditions} that must be solved in order to find the symmetries are generally undetermined, making the unknown functions defining the generators of symmetries, referred to as the \textit{infinitesimals}~\cite{bluman1989symmetries}, difficult to find. Typically, the equations must therefore be solved using ans\"{a}tze for the infinitesimals but in general the structure of the infinitesimals is unknown rendering the task of designing the ans\"{a}tze challenging.

For a single first order ODE, Cheb--Terrab and Kolokolnikov~\cite{cheb2003first} have designed a set of ans\"{a}tze for the infinitesimals that are capable of finding symmetries for a large class of models. Moreover, certain systems of first order ODEs can be formulated as a single higher order ODE~\cite{harrington2017reduction}, for which the linearised symmetry condition decomposes and can be systematically solved for the infinitesimals. Examples where the latter strategy is applicable are the Lorenz model~\cite{sen1990lie} and various models describing disease transmission in epidemiology~\cite{nucci2005using}. Nonetheless, many systems of first order ODEs cannot be cast as a single higher order ODE and accordingly symmetries of these systems cannot be determined using this methodology.

Recently, it was shown by Ohlsson et al.~\cite{ohlsson2022geometry} that it is possible to extend the symmetries of the single phase plane ODE corresponding to a system of two autonomous ODEs to symmetries of the time-dependent system by solving the so-called \textit{lifting condition}. If the symmetry condition for the phase plane ODE is easier to solve than those of the original system, this connection provides an alternative approach for determining time domain symmetries of autonomous two-dimensional systems.

However, finding symmetries of a dynamical system is only half the battle if the goal is to improve our understanding of it. In order to provide meaningful information about the structure of the system and a powerful way to represent its dynamics, the ability to interpret the action of the symmetries is also paramount. 


In this paper, we address the challenges of finding symmetries and using them to derive insight into the dynamics of the particular class of models where the phase plane ODE is separable. Specifically, we present two main theoretical results establishing the existence of two non-trivial symmetry generators of \textit{any} separable phase plane ODE (Theorem \ref{thm:separable}) and their interpretations in terms of translations of the internal energy of phase plane trajectories (Theorem \ref{thm:interpretation_separable_symmetries}). Using the methodology developed in~\cite{ohlsson2022geometry} we also provide the explicit form of the lifting condition for these symmetries. The symmetries  are connected to the Hamiltonian structure of the dynamical system through their action on the space of solutions; raising or lowering the energy of a trajectory.

We exemplify the analysis based on our theoretical developments using three biological models: a canonical biological oscillator, the Lotka--Volterra (LV) model describing predator-prey dynamics, and the SIR model describing the spread of a disease in a population. We construct the non-trivial energy translation symmetries and provide analytic solutions for the corresponding lifting conditions, extending them to incorporate the temporal dynamics of the models. Furthermore, we discuss the action of the symmetries on biologically meaningful observables and the structure of the space of solutions.
\section{Orthogonal phase plane symmetries of separable models that act as translations in the internal energy}
\label{sec:orthogonal_symmetries}
Consider the autonomous two-state model
\begin{equation}
\label{eq:sys_auto}
    \frac{\mathrm{d}u}{\mathrm{d}t}=\dot{u}=\omega_u(u,v)\,,\quad\frac{\mathrm{d}v}{\mathrm{d}t}=\dot{v}=\omega_v(u,v)\,,
\end{equation}
and its corresponding $(u,v)$ phase plane representation given by
\begin{equation}
\label{eq:sys_phase_plane}
    \frac{\mathrm{d}v}{\mathrm{d}u}=v'=\Omega(u,v)=\frac{\omega_v(u,v)}{\omega_u(u,v)},\quad\omega_u(u,v)\neq 0.
\end{equation}
For separable ODEs the \textit{reaction terms} are of the form
\begin{equation}
\label{eq:sep_reac}
    \omega_u(u,v)=f_u(u)g_u(v)h(u,v)\,,\quad\,\omega_v(u,v)=f_v(u)g_v(v)h(u,v) ,
\end{equation}
where $f_u,f_v,g_u,g_v,h$ are continuous functions, and for $f_u(u)\neq 0\,, g_v(v)\neq 0$ the phase plane ODE in Eq.~\eqref{eq:sys_phase_plane} can be directly integrated to give
\begin{equation}
\label{eq:internal_energy}
    H=\int\frac{g_u(v)}{g_v(v)}\mathrm{d} v-\int\frac{f_v(u)}{f_u(u)}\mathrm{d} u .
\end{equation}
The arbitrary integration constant $H$ parametrises the space of solutions of Eq.~\eqref{eq:sys_phase_plane} and is often interpreted as the \textit{internal energy} of a trajectory since it is a first integral of Eq.~\eqref{eq:sys_auto} and therefore conserved under time evolution of the original system.

\subsection{Symmetries in the time domain}
Consider the space $M_3$ parametrised by $(t,u,v)$, which we will refer to as the \textit{time domain}, and let $\Gamma_{3,\epsilon} : M_3 \to M_3$ be a family of Lie point transformations parametrised by $\epsilon$, and generated by the vector field $X = \xi(t,u,v)\partial_t + \eta_u(t,u,v)\partial_u + \eta_v(t,u,v)\partial_v$~\cite{bluman1989symmetries,hydon2000symmetry}. A transformation $\Gamma_{3,\epsilon}$ constitutes a \textit{symmetry} of the system in Eq.~\eqref{eq:sys_auto} if $X$ satisfies the \textit{linearised symmetry conditions}~\cite{bluman1989symmetries,hydon2000symmetry,olver2000applications,stephani1989differential}:
\begin{equation}
\label{eqn:sym_con_ODE}
    \left. X^{(1)}\left( \dot{u} - \omega_u \right)\right|_{\dot{u} =\omega_u,\dot{v} =\omega_v} = 0\, \,,\quad \left. X^{(1)}\left( \dot{v} - \omega_v \right)\right|_{\dot{u} =\omega_u,\dot{v} =\omega_v} = 0\,,
\end{equation}
where the first prolongation $X^{(1)}$ of the generator is given by
\begin{equation}
\label{eqn:prolonged_vector_J5}
    X^{(1)} = X + \eta_u^{(1)}(t,u,v)\partial_{\dot{u}} + \eta_v^{(1)}(t,u,v)\partial_{\dot{v}} , \quad 
\end{equation} 
and the prolonged infinitesimals $\eta_u^{(1)},\eta_v^{(1)}$ are calculated using the total derivative $D_t = \partial_t + \dot{u}\partial_u + \dot{v}\partial_v$ according to
\begin{equation}
\label{eqn:prolonged_infinitessimals}
    \eta_u^{(1)} = D_t\eta_u - \dot{u}D_t\xi \,, \quad \eta_v^{(1)} = D_t\eta_v - \dot{v}D_t\xi\,.
\end{equation}
The symmetry $\Gamma_{3,\epsilon}$ is \textit{non-trivial} if it acts non-trivially on the space of solutions to Eq.~\eqref{eq:sys_auto}.

\subsection{Symmetries in the phase plane and their lift to the time domain}
Now, we consider the ODE in Eq.~\eqref{eq:sys_phase_plane} where the variables $u$ and $v$ parameterise the two-dimensional phase plane $M_2$. In analogy with the time domain, let $\Gamma_{2,\epsilon} : M_2 \to M_2$ be a Lie point transformation which is generated by the vector field $Y=\zeta_u(u,v)\partial_u + \zeta_v(u,v)\partial_v$. Then, $\Gamma_{2,\epsilon}$ is a symmetry of the phase plane ODE in Eq.~\eqref{eq:sys_phase_plane} if it satisfies the linearised symmetry condition
\begin{equation}
   \left. Y^{(1)}\left( v' - \Omega \right) \right|_{v' =\Omega} = 0 ,
\label{eq:sym_con_phase}    
\end{equation}
where the the first prolongation of the infinitesimal phase plane generator is given by
\begin{equation}
    Y^{(1)} = Y + \zeta_v^{(1)}(u,v)\partial_{v'}\,,
    \label{eq:prolonged_vector_J3}
\end{equation}
and
\begin{equation}
    \zeta_v^{(1)} = D_u\zeta_v - v'D_u\zeta_u\,,
\end{equation}
is the prolonged infinitesimal defined by the total derivative $D_u = \partial_u + v'\partial_v\,$ in the $(u,v)$ phase plane. As in the time domain, a symmetry transformation $\Gamma_{2,\epsilon}$ is non-trivial if it acts non-trivially on the space of solutions to Eq.~\eqref{eq:sys_phase_plane}.

The action of a phase plane symmetry generated by $Y$ is made manifest by introducing the \textit{canonical coordinates}~\cite{hydon2000symmetry} $(s,r)$ in the phase plane, defined by $Ys = 1$ and $Yr = 0$, in terms of which the transformation $\Gamma_{2,\epsilon}$ becomes
\begin{equation}
    \Gamma_{2,\epsilon}:\left(s,r\right)\mapsto\left(s+\epsilon,r\right).
    \label{eq:symmetry_canonical}
\end{equation}
The two canonical coordinates $(s,r)$ can be interpreted as properties of the model in Eq.~\eqref{eq:sys_phase_plane} that are \textit{changed} and \textit{conserved}, respectively, under the action of the symmetry. In fact, $r$ is an example of a \textit{differential invariant} of the generator $Y$~\cite{bluman1989symmetries,hydon2000symmetry}.

Recently, it was shown by Ohlsson et al.~\cite{ohlsson2022geometry} that it is possible to lift an infinitesimal generator $Y$ of a symmetry in the phase plane to an infinitesimal generator $X$ of a corresponding symmetry $\Gamma_{3,\epsilon}$ in the time domain. More precisely, the infinitesimal generator is given by $X=\xi(t,u,v)\partial_t+Y$ where the time infinitesimal $\xi$ solves the \textit{lifting condition}~\cite{ohlsson2022geometry}
\begin{equation}
\label{eq:lifting_condition}
    \partial_t \xi + \omega_u \partial_u \xi + \omega_v \partial_v \xi =  \frac{1}{\omega_u} \left( \vphantom{\frac{1}{\omega_u}} \left( \omega_u\partial_u + \omega_v\partial_v \right) \zeta_u - \left( \zeta_u\partial_u + \zeta_v\partial_v \right) \omega_u \right) \,.
\end{equation}

\subsection{Two non-trivial phase plane symmetries of separable models}
\label{ssec:separable_symmetries}
Our main theoretical result, established through the following two theorems, is the construction of two orthogonal, non-trivial symmetries for any separable model in the $(u,v)$ phase plane. 
\begin{theorem}[\textbf{Orthogonal phase plane symmetries of separable ODEs}]
Let the functions $f_u,f_v,g_u,g_v$ be given by Eq.~\eqref{eq:sep_reac}. Then the vector fields 
\begin{align}
    Y_u&=\frac{f_u(u)}{f_v(u)}\partial_u,\quad f_v(u)\neq 0,\label{eq:Y_u}\\
    Y_v&=\frac{g_v(v)}{g_u(v)}\partial_v,\quad g_u(v)\neq 0.\label{eq:Y_v}
\end{align}
generate two orthogonal symmetries of the phase plane ODE in Eq.~\eqref{eq:sep_reac}.
\label{thm:separable}
\end{theorem}
\begin{proof}
The linearised symmetry condition in Eq.~\eqref{eq:sym_con_phase} is equivalent to
\begin{equation}
\label{eqn:sym_con_phase_comp}
    \omega_u \left( \left( \omega_u\partial_u + \omega_v\partial_v \right) \zeta_v - \left( \zeta_u\partial_u + \zeta_v\partial_v \right) \omega_v \right) = \omega_v \left( \left( \omega_u\partial_u + \omega_v\partial_v \right) \zeta_u - \left( \zeta_u\partial_u + \zeta_v\partial_v \right) \omega_u \right) .
\end{equation}
A straightforward calculation shows that the infinitesimals $\zeta_u(u,v)=f_u(u)/f_v(u)$, $\zeta_v(u,v)=0$ of $Y_u$ and $\zeta_u(u,v)=0$, $\zeta_v(u,v)=f_v(v)/f_u(v)$ of $Y_v$ both satisfy Eq.~\eqref{eqn:sym_con_phase_comp}, and the generators $Y_u$ and $Y_v$ are clearly orthogonal.
\end{proof}

It is well-known that separable first order ODEs have a unidirectional symmetry in the independent variable~\cite{cheb2003first,stephani1989differential}. However, we are free to choose how to parameterise the phase plane by considering either $u$ or $v$ as the independent variable, and the symmetries of the phase plane ODE are independent of this choice~\cite{ohlsson2022geometry}. Based on this observation, Theorem 1 extends the results of~\cite{cheb2003first,stephani1989differential} to include a symmetry generator $Y_v$ in the dependent coordinate direction.

\begin{theorem}[\textbf{Non-triviality of phase plane symmetries of separable ODEs}]
\label{thm:interpretation_separable_symmetries}
The symmetry transformations $\Gamma_{2,\epsilon}^{u}$ and $\Gamma_{2,\epsilon}^{v}$ generated by $Y_u$ in Eq.~\eqref{eq:Y_u} and $Y_v$ in Eq.~\eqref{eq:Y_v}, respectively, are non-trivial and act on the internal energy $H$ in Eq.~\eqref{eq:internal_energy} according to
\begin{equation}
\label{eq:Gamma_u_v}
    \Gamma_{2,\epsilon}^{u}:H \mapsto H-\epsilon, \quad \Gamma_{2,\epsilon}^{v}:H \mapsto H+\epsilon.
\end{equation}
\end{theorem}
\begin{proof}
The canonical coordinates of $\Gamma_{2,\epsilon}^{u}$ and $\Gamma_{2,\epsilon}^{v}$, respectively, are given by
\begin{equation}
    s_u=\int\frac{f_v(u)}{f_u(u)}\mathrm{d} u,\quad r_u=v, \quad \mathrm{and} \quad 
    s_v=\int\frac{g_u(v)}{g_v(v)}\mathrm{d} v,\quad r_v=u.
\end{equation}
Consequently, the internal energy $H$ in Eq.~\eqref{eq:internal_energy} of any phase plane trajectory can be expressed as
  \begin{equation}
      H=s_v-s_u\,,
      \label{eq:internal_energy_canonical}
  \end{equation}
and the claims of the theorem follow since $H$ parametrises the solution space of Eq.~\eqref{eq:sys_phase_plane}.
\end{proof}

According to Eq.\eqref{eq:Gamma_u_v} the orthogonal symmetries act as translations of the internal energy and are consequently related to the Hamiltonian formulation of separable models. In particular, there is a single family of solutions parameterised by the energy, and $\Gamma^u_{2,\epsilon}$ and $\Gamma^u_{2,\epsilon}$ act on this one-dimensional space by raising and lowering the energy.

For separable models and the non-trivial symmetry generators in Eqs.~\eqref{eq:Y_u} and \eqref{eq:Y_v} the lifting condition in Eq.~\eqref{eq:lifting_condition} generally simplifies. In particular, for the case $h(u,v)=1$ which is common in applications, the lifting condition reduces to
\begin{equation}
    \partial_t \xi + \omega_u \partial_u \xi + \omega_v \partial_v \xi =  f_u(u) \partial_u \! \left( \frac{1}{f_v(u)} \right) + g_v(v) \partial_v \! \left( \frac{1}{g_u(v)} \right) \,.
\end{equation}

\section{Application to biological models}
We will now use the phase plane symmetries discussed in the previous section to analyse three separable models in biology: a canonical oscillator model~\cite{murray2002}, the Lotka--Volterra (LV) model~\cite{lotka1920undamped, lotka1925elements, volterra1926variations}, and the epidemiological SIR model~\cite{SIR-model}. In all examples, we extract the orthogonal symmetry generators in the phase plane and use them to describe the action on biologically relevant quantities characterising the solution trajectories. Furthermore, we compute the lift to the time domain which gives non-trivial symmetries of the original two-state model that act by translations on the internal energy and are therefore biologically interpretable. All of the plots can be regenerated using the open-source github repository associated with this article (see \url{https://github.com/JohannesBorgqvist/separable_phase_plane_symmetries}).

\subsection{A canonical oscillator model}
Let $u(t)$ and $v(t)$ be the states at time $t$ of the canonical oscillator model whose dynamics is governed by~\cite{murray2002}
\begin{equation}
        \frac{\mathrm{d}u}{\mathrm{d}t}=u\left(1-\sqrt{u^2+v^2}\right)-\lambda v\,,\quad\frac{\mathrm{d}v}{\mathrm{d}t}=v\left(1-\sqrt{u^2+v^2}\right)+\lambda u,
    \label{eq:biological_oscillator}
\end{equation}
where $\lambda$ is a positive constant describing the angular frequency of the oscillator. In terms of the polar coordinates $(\sigma,\theta)$ defined by
\begin{equation}
    u=\sigma\cos\theta\,,\quad v=\sigma\sin\theta\,,
\label{eq:polar_coordinates}
\end{equation}
the model becomes separable 
\begin{equation}
\label{eq:oscillator}
    \frac{\mathrm{d}\theta}{\mathrm{d}t}=\lambda,\quad\frac{\mathrm{d}\sigma}{\mathrm{d}t}=\sigma(1-\sigma),
\end{equation}
and the corresponding $(r,\theta)$ phase plane ODE 
\begin{equation}
\label{eq:oscillator_phase}
    \frac{\mathrm{d}\sigma}{\mathrm{d}\theta}=\frac{\sigma(1-\sigma)}{\lambda},
\end{equation}
can be solved to produce the internal energy $H_{\mathrm{Osc}} = \ln \sigma-\ln|1-\sigma|-\theta/\lambda$. 

The two orthogonal phase plane symmetries in Eq.~\eqref{eq:Gamma_u_v} are generated by
\begin{equation}
Y_{\theta}^{\mathrm{Osc}}=\lambda\partial_{\theta},\quad Y_{\sigma}=\sigma(1-\sigma)\partial_\sigma,
\label{eq:oscillator_symmetries}
\end{equation}
with corresponding canonical coordinates $(s_\theta,r_\theta)=(\theta/\lambda,\sigma)$ and $(s_\sigma,r_\sigma)=(\ln\sigma - \ln|1-\sigma|,\theta)$. In addition to the action on the internal energy $H_{\mathrm{Osc}}$ in Eq.~\eqref{eq:Gamma_u_v}, the integral curves for the generators can also be used to describe the transformations of the states themselves
\begin{equation}
    \Gamma_{2,\epsilon}^{\mathrm{Osc},\theta}: (\theta,\sigma) \mapsto \left(\theta + \lambda\epsilon,\sigma\right)\,,
\end{equation}
\begin{equation}
    \Gamma_{2,\epsilon}^{\mathrm{Osc},\sigma}: (\theta,\sigma) \mapsto \left(\theta,\frac{1}{1+\left(\frac{1}{\sigma}-1\right)e^{-\epsilon}}\right)\,.
\end{equation}

All solutions of the oscillator model have the same frequency and are related by a constant shift in the angular coordinate $\theta$, or equivalently a translation in time, implying that the only qualitative distinction between solutions is whether they lie inside ($\sigma<1$) or outside ($\sigma>1$) the limit cycle (see \cite{ohlsson2022geometry} for an in-depth discussion). Consequently, the action of the symmetries $\Gamma_{2,\epsilon}^{\mathrm{Osc},\theta}$ and $\Gamma_{2,\epsilon}^{\mathrm{Osc},\sigma}$ is a rotation of the phase plane trajectory, or equivalently, a reparametrisation of the phase plane.

This fact is also reflected in the time domain symmetries obtained by solving the lifting condition in Eq.~\eqref{eq:lifting_condition}
\begin{align}
    X_\theta^{\mathrm{Osc}}&=F_{\theta}\left(H_{\mathrm{Osc}}\right)\partial_t+\lambda\partial_\theta\,,\label{eq:osc_lift_angle}\\X_\sigma^{\mathrm{Osc}}&=F_{\sigma}\left(H_{\mathrm{Osc}}\right)\partial_t+\sigma(1-\sigma)\partial_\sigma\,,\label{eq:osc_lift_radius}
\end{align}
where $F_{\theta}$ and $F_{\sigma}$ are two arbitrary differentiable functions. The time infinitesimals are constant on each trajectory making the equivalence to shifting solutions in time manifest.

The solutions and the actions of the symmetries in the phase plane and time domain for the biological oscillator model are illustrated in Fig. \ref{fig:oscillator}.

\begin{figure}[ht!]
    \centering
    \includegraphics{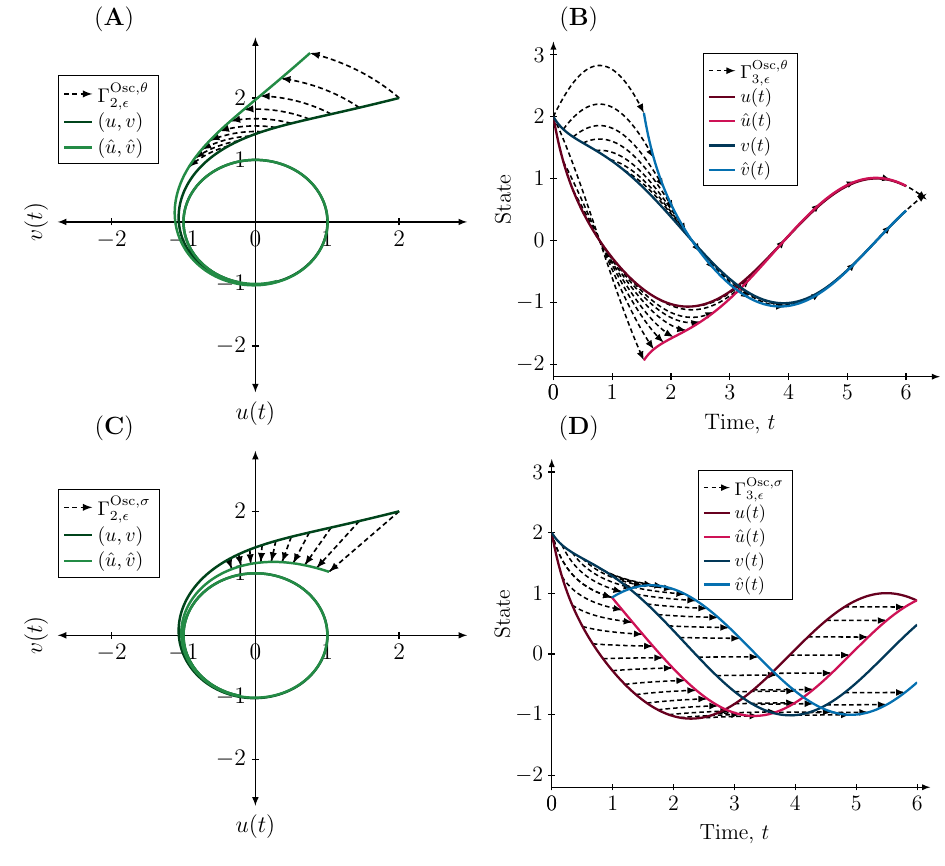}
    \caption{\textit{Symmetries of a canonical oscillator model}. Original and transformed solution curves with $\lambda=1$ illustrated in four cases: (\textbf{A}) the symmetry $\Gamma^{\mathrm{Osc},\theta}_{2,\epsilon}$ in the $(u,v)$ phase plane with $\epsilon=\pi/6$ and (\textbf{B}) the symmetry $\Gamma^{\mathrm{Osc},\theta}_{3,\epsilon}$ with $F_{\theta}(x)=1$ and $\epsilon=\pi/6$ in the time domain, (\textbf{C}) the symmetry $\Gamma^{\mathrm{Osc},\sigma}_{2,\epsilon}$ in the $(u,v)$ phase plane with $\epsilon=0.75$, (\textbf{D}) the symmetry $\Gamma^{\mathrm{Osc},\sigma}_{3,\epsilon}$ in the time domain with $F_{\sigma}(x)=1$ and  $\epsilon=0.75$. }
    \label{fig:oscillator}
\end{figure}

\subsection{The Lotka--Volterra model}
The dimensionless LV model~\cite{lotka1920undamped, lotka1925elements, volterra1926variations} is given by
\begin{equation}
\frac{\mathrm{d}u}{\mathrm{d}t}=u(1-v),\quad\frac{\mathrm{d}v}{\mathrm{d}t}=\alpha v(u-1),\quad u(t),v(t)\geq 0,
\label{eq:LV}
\end{equation}
where $u(t)$ and $v(t)$ correspond, respectively, to the number of prey and predators in the population at time $t$, and $\alpha$ is a rate parameter. The corresponding $(u,v)$ phase plane ODE
\begin{equation}
\dfrac{\mathrm{d}v}{\mathrm{d}u}=\dfrac{\alpha v(u-1)}{u(1-v)},
  \label{eq:LV_phase_plane}
\end{equation}
is directly solvable, yielding the internal energy\footnote{To conform to the standard treatment of the Lotka--Volterra model we have used the ambiguity in the definition of the integration constant in Eq.~\eqref{eq:internal_energy} to redefine the internal energy to be a manifestly positive quantity. Consequently, the action of the generators $Y_u^{\mathrm{LV}}$ and $Y_v^{\mathrm{LV}}$ on $H_{\mathrm{LV}}$ in Eq.~\eqref{eq:Gamma_u_v} are reversed.} $H_{\mathrm{LV}}=\alpha (u-\ln u)+v-\ln v$. 

The two phase plane symmetry generators in Eqs.~\eqref{eq:Y_u} and \eqref{eq:Y_v} are given by
\begin{equation}
\label{eq:Y_LV}
    Y^{\mathrm{LV}}_u=\dfrac{1}{\alpha}\left(\dfrac{u}{u-1}\right)\partial_u,\quad Y^{\mathrm{LV}}_v=\dfrac{v}{1-v}\partial_v,
\end{equation}
with corresponding canonical coordinates $(s_u,r_u)=(\alpha (u-\ln u),v)$ and $(s_v,r_v)=(\ln v-v,u)$. The point transformations in the $(u,v)$ phase plane generated by $Y_u^{\mathrm{LV}}$ and $Y_v^{\mathrm{LV}}$ are found through the integral curves as
\begin{equation}
\label{eq:Gamma_u_LV}
    \Gamma_{2,\epsilon}^{\mathrm{LV},u}: (u,v) \mapsto \left(-W\left[-\exp\left(-\frac{\epsilon}{\alpha} + \ln u - u\right)\right],v\right) \,,
\end{equation}
\begin{equation}
\label{eq:Gamma_v_LV}
    \Gamma_{2,\epsilon}^{\mathrm{LV},v}: (u,v) \mapsto \left(u,-W\left[-\exp\left(\epsilon + \ln v - v\right)\right]\right) \,,
\end{equation}
where $W$ is the Lambert W function (see~\cite{corless1996lambertw,lehtonen2016Lambert} for recent overviews of the Lambert W function and its uses in biological modelling). The properties of $W$ determine the domain and range of the transformations in Eqs.~\eqref{eq:Gamma_u_LV} and \eqref{eq:Gamma_v_LV}. In particular, the integral curves are restricted by the singularities in $Y_u^{\mathrm{LV}}$ and $Y_v^{\mathrm{LV}}$ in Eq.~\eqref{eq:Y_LV}, meaning that for a fixed transformation parameter $\epsilon$ there is either a segment of the original solution where the transformation is not defined or a segment of the transformed solution which is not reached by the transformation (see Fig.~\ref{fig:LV}).

In contrast to the previous case of the biological oscillator, not all solutions to the LV model are related through time translations. Instead, the inequivalent solutions can be characterized by the minimum and maximum populations of the species $u$ and $v$. These extrema, corresponding to the intersection of the solution trajectory and the phase plane nullclines $u=1$ and $v=1$, can be expressed in terms of the internal energy as
\begin{equation}
    u_{\min} = -W_0\left[ -\exp\left( \frac{1-H_{\mathrm{LV}}}{\alpha}\right)\right] \,, \quad u_{\max} = -W_{-1}\left[ -\exp\left( \frac{1-H_{\mathrm{LV}}}{\alpha}\right)\right] \,,
\end{equation}
\begin{equation}
    v_{\min} = -W_0\left[ -\exp\left( \alpha-H_{\mathrm{LV}} \right)\right] \,, \quad v_{\max} = -W_{-1}\left[ -\exp\left( \alpha-H_{\mathrm{LV}} \right)\right] \,,
\end{equation}
where $W_0$ and $W_{-1}$ are the real branches of the Lambert W function. Consequently, the action in Eq.~\eqref{eq:Gamma_u_v} provides a direct way to interpret the action of $\Gamma_u^{\mathrm{LV}}$ and $\Gamma_v^{\mathrm{LV}}$ in terms of biologically meaningful quantities.

Lifting the phase plane symmetry generators of the LV model to the time domain yields
\begin{align}
X^{\mathrm{LV}}_u &= - \left( \int_{u_0}^u \frac{\mathrm{d}z}{\alpha(z-1)^2 \left( 1 + W\left[-\exp\left(\alpha\left(z-\ln z\right) - H_{\mathrm{LV}}\right)\right] \right)} \right) \partial_t + \frac{1}{\alpha} \left( \dfrac{u}{u-1} \right) \partial_u \,,\\
X^{\mathrm{LV}}_v &= - \left( \int_{v_0}^v \frac{\mathrm{d}z}{\alpha(z-1)^2 \left( 1 + W\left[-\exp\left(\frac{1}{\alpha} \left(z-\ln z - H_{\mathrm{LV}}\right)\right)\right]\right)} \right) \partial_t + \frac{v}{1-v}\partial_v \,.
\end{align}
The time infinitesimals encode the local transformations in time required to construct a map between two local segments of the solution trajectories. In contrast to the oscillator model this transformation is not simply a constant translation for the entire solution trajectory, but depends non-trivially on the states. The fact that this symmetry is not obvious from the original time domain formulation in Eq.~\eqref{eq:LV} illustrates the power of the phase plane symmetry analysis.

The solutions and actions of the symmetries in the phase plane and the time domain for the LV model are illustrated in Fig. \ref{fig:LV}.


\begin{figure}[htbp!]
    \centering
    \includegraphics{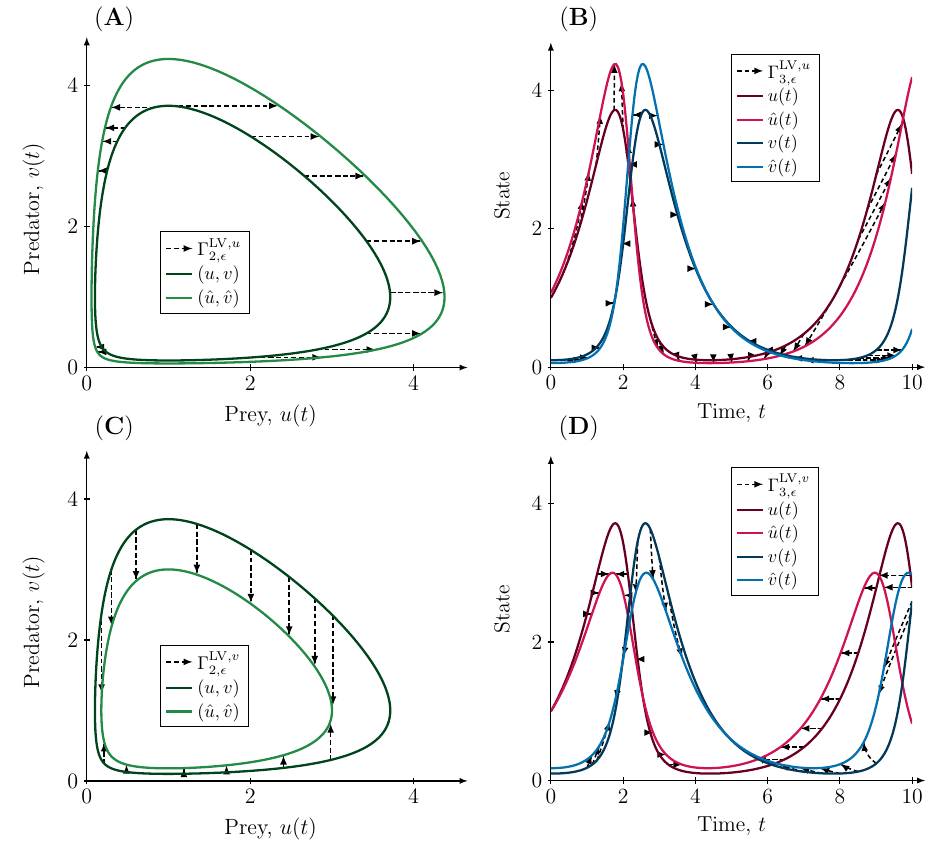}
    \caption{\textit{Symmetries of the LV model}. The original solution curves defined by $\alpha=1$ are transformed with $\epsilon=0.5$ in order to produce transformed solution curves which are illustrated in four cases: (\textbf{A}) the symmetry $\Gamma^{\mathrm{LV},u}_{2,\epsilon}$ in the $(u,v)$ phase plane, (\textbf{B}) the symmetry $\Gamma^{\mathrm{LV},u}_{3,\epsilon}$ in the time domain, (\textbf{C}) the symmetry $\Gamma^{\mathrm{LV},v}_{2,\epsilon}$ in the $(u,v)$ phase plane, and (\textbf{D}) the symmetry $\Gamma^{\mathrm{LV},v}_{3,\epsilon}$ in the time domain.}
    \label{fig:LV}
\end{figure}

\subsection{The SIR model}
The original SIR model~\cite{SIR-model} describes the spread of a disease in a population and can be reduced to a dimensionless model for the sub-populations of susceptible, $S(t)$, and infected, $I(t)$, individuals at time $t$
\begin{equation}
  \frac{\mathrm{d}S}{\mathrm{d}t}=-SI\,,\quad\frac{\mathrm{d}I}{\mathrm{d}t}=I(S-\rho)\,,\quad I(t),S(t)\geq 0\,,
  \label{eq:SIR}
\end{equation}
where $\rho$ is a recovery rate parameter. The corresponding $(S,I)$ phase plane ODE
\begin{equation}
    \frac{\mathrm{d}I}{\mathrm{d}S}=\frac{\rho-S}{S}\,,
\label{eq:SI_phase_plane}
\end{equation}
can be integrated according to Eq.~\eqref{eq:internal_energy} to an expression for the internal energy $H_{\mathrm{SIR}} = I + S - \rho\ln S$.

The two phase plane symmetry generators in Eqs.~\eqref{eq:Y_u} and \eqref{eq:Y_v} are given by
\begin{equation}
  Y^{\mathrm{SIR}}_S=\left(\frac{S}{\rho-S}\right)\partial_S \,, \quad Y^{\mathrm{SIR}}_I=\partial_I \,,
  \label{eq:SIR_symmetries}
\end{equation}
with canonical coordinates $(s_S,r_S)=(\rho\ln S-S,I)$ and $(s_I,r_I)=(I,S)$, and the corresponding point transformations are
\begin{equation}
\label{eq:Gamma_S_SIR}
    \Gamma_{2,\epsilon}^{\mathrm{SIR},S}: (S,I) \mapsto \left(-\rho W\left[-\frac{1}{\rho}\exp\left(\frac{1}{\rho}\left(\epsilon + \rho \ln S - S\right)\right)\right] , I\right) \,,
\end{equation}
\begin{equation}
\label{eq:Gamma_I_SIR}
    \Gamma_{2,\epsilon}^{\mathrm{SIR},I}: (S,I) \mapsto \left(S,I+\epsilon\right) \,.
\end{equation}
Just as for the LV model, the limitations on the range and domain for the Lambert W function implies that the transformation $\Gamma_{2,\epsilon}^{\mathrm{SIR},S}$ is not one-to-one for the entire solution trajectory and its transform. In addition, while there is no singularity in the integral curve for the generator $Y_I^{\mathrm{SIR}}$, the transformation $\Gamma_{2,\epsilon}^{\mathrm{SIR},I}$ must clearly be restricted in order for the states to remain in the physical range.

The solution trajectory in Eq.~\eqref{eq:SI_phase_plane} has a maximum at $S=\rho$ given by
\begin{equation}
    I_{\max}=H_{\mathrm{SIR}}-\left(\rho-\rho\ln \rho \right),
  \end{equation}
which implies that the symmetries generated by $Y_S$ and $Y_I$ can be interpreted as, respectively, decreasing and increasing the maximum number of infected individuals. In addition, the maximum and minimum of susceptible individuals is obtained from the intersection with the stable locus $I=0$ as
\begin{equation}
    S_{\min} = -\rho W_0 \left[ -\frac{1}{\rho} \exp\left( -\frac{H_{\mathrm{SIR}}}{\rho}\right) \right]\,, \quad S_{\max} = -\rho W_{-1} \left[ -\frac{1}{\rho} \exp\left( -\frac{H_{\mathrm{SIR}}}{\rho}\right) \right] \,.
\end{equation}
Similarly to the case for the LV model, these expressions relate the translation in internal energy $H_{\mathrm{SIR}}$ to quantities of direct biological relevance.

Solving the lifting condition for $Y^{\mathrm{SIR}}_S$ and $Y^{\mathrm{SIR}}_I$, we obtain the following vector fields as generators of symmetries in the time domain 
\begin{align}
X^{\mathrm{SIR}}_S &= - \left( \int_{S_0}^S \frac{\mathrm{d}z}{(\rho-z)^2(H_{\mathrm{SIR}}-z+\rho\ln z)} \right) \partial_t + \left(\frac{S}{\rho-S}\right) \partial_S \,,\\
X^{\mathrm{SIR}}_I &= - \left( \bigintsss_{I_0}^I \frac{\mathrm{d}z}{\rho z^2 \left(1+W\left[ - \frac{1}{\rho} \exp \left( \frac{1}{\rho} (z  -H_{\mathrm{SIR}}) \right)\right]\right)} \right) \partial_t+\partial_I,\label{eq:SIR_I_lift}
\end{align}
which as in case of the LV model contain time infinitesimals which are highly non-trivial to deduce directly from the time domain system in Eq.~\eqref{eq:SIR}.

The action of the symmetries of the SIR model in both the phase plane and the time domain is illustrated in Fig. \ref{fig:SIR}. 

\begin{figure}[ht!]
    \centering
    \includegraphics{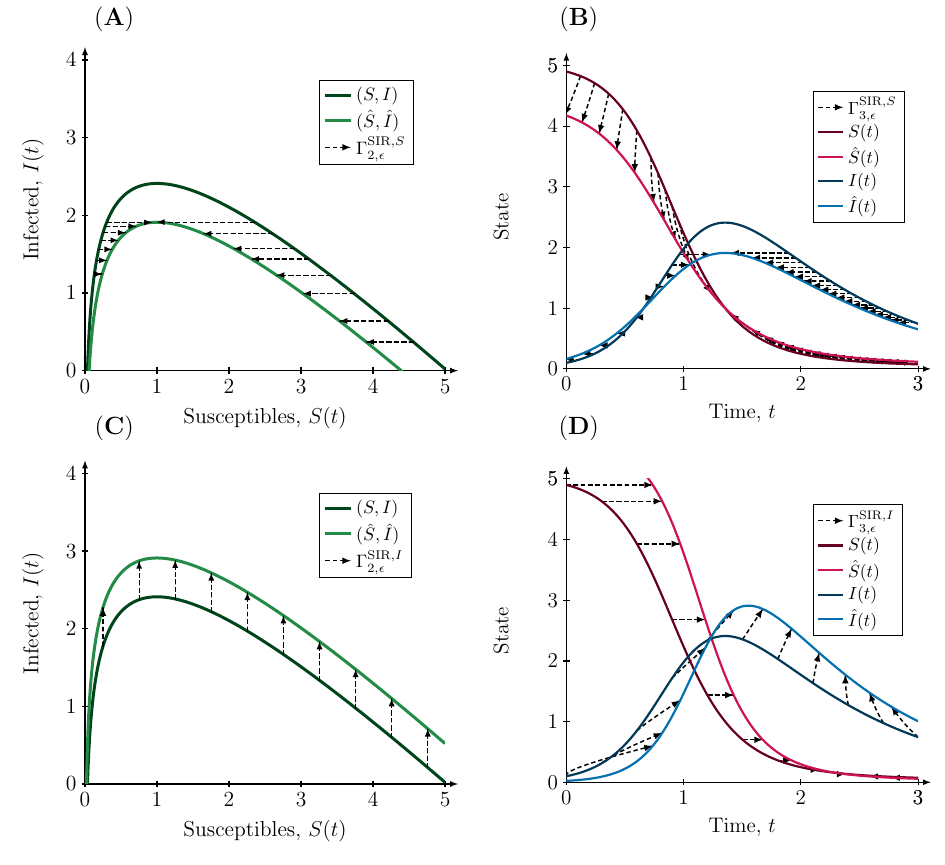}
    \caption{\textit{Symmetries of the SIR model}. The original solution curves defined by $\rho=1$ are transformed with $\epsilon=0.5$ in order to produce transformed solution curves which are illustrated in four cases: (\textbf{A}) the symmetry $\Gamma^{\mathrm{SIR},S}_{2,\epsilon}$ in the $(S,I)$ phase plane, (\textbf{B}) the symmetry $\Gamma^{\mathrm{SIR},S}_{3,\epsilon}$ in the time domain, (\textbf{C}) the symmetry $\Gamma^{\mathrm{SIR},I}_{2,\epsilon}$ in the $(S,I)$ phase plane and (\textbf{D}) the symmetry $\Gamma^{\mathrm{SIR},I}_{3,\epsilon}$ in the time domain.}
    \label{fig:SIR}
\end{figure}

\section{Discussion and conclusions}
The strength of this work lies in (i) the development of a methodology for finding energy translation symmetries, and (ii) the interpretation of the action of these symmetries. On the latter point, the theoretical result showing that the identified symmetries translate the internal energy of solution curves in the phase plane enables us to interpret their action in terms of physically meaningful quantities. We exemplify this result using three different biological models, where we demonstrate what it means to alter the internal energy in practice. For example, for the SIR model we show that altering the internal energy amounts to changing the maximum number of infected individuals during an epidemic. On the first point, the technique we present for finding energy translation symmetries constitutes a straightforward methodology for solving the normally undetermined linearised symmetry conditions. It is well-known within the analysis of differential equations that the dimension of the symmetry group of systems of first order ODEs is infinite (see~\cite{nucci2008lie} for a discussion on this). The challenge, then, is to find informative symmetries in the absence of a general strategy for solving the underdetermined linearised symmetry conditions. In the case of autonomous models, it is always possible to reformulate the original system consisting of, say, $n$ equations as a system of $n-1$ equations by considering one of the states as the independent variable (see the discussion of complete symmetry groups of second order systems in~\cite{nucci1996complete}). Using this result, combined with the recently derived lifting condition~\cite{ohlsson2022geometry}, we demonstrate how the action of the energy translation symmetries in the phase plane is realised in the time domain. Consequently, for the class of models defined by Eqs.~\eqref{eq:sys_auto} and \eqref{eq:sep_reac}, we have proposed and implemented a straightforward methodology for solving the normally difficult problem of finding Lie symmetries in the time domain. Nevertheless, it is important to emphasise that this methodology is restricted to this particular type of system of first order ODEs.

The main limitation of the methodology in the present work is that the scope is restricted to autonomous models with separable reaction terms. Autonomy gives rise to a closed-form phase plane representation. Separability then allows us to immediately integrate the phase plane equation, and obtain two symmetries that can be readily understood in terms of altering the internal energy of the system. However, there are numerous examples, with applications in both physics and biology, of systems that are not immediately separable.
In such cases, we must first find phase plane symmetries by solving the linearised symmetry condition for the phase plane ODE in Eq.~\eqref{eq:sym_con_phase}. As discussed previously, solving this linearised symmetry condition remains a challenging problem in general as the problem is ill-posed.

In light of this difficulty, an interesting extension of this work would be to attempt to solve the linearised symmetry condition for the case of non-separable reaction terms. The most straightforward approach is to design ans\"{a}tze for the unknown infinitesimals $\zeta_u$ and $\zeta_v$ solving the linearised symmetry condition in Eq.~\eqref{eq:sym_con_phase}. A set of ans\"{a}tze for these infinitesimals for a large class of single first order ODEs has been previously designed~\cite{cheb2003first}, which serves as a natural staring point for this endeavour. Given a solution of the linearised symmetry condition for the phase plane ODE, the analysis presented in this work can be readily extended to the case of non-separable reaction terms. Moreover, the Hamiltonian structure of, e.g., the LV model has been considered before~\cite{nutku1990hamiltonian}. In future work, we intend to extend such an analysis to general separable systems based on the results presented in the present paper. In addition, an interesting avenue of investigation is to consider the algebraic structure of the space of solutions endowed by the symmetry generators $X_u$ and $X_v$ in Eqs.~\eqref{eq:Y_u} and~\eqref{eq:Y_v}.

In order to show how the theoretical results developed in Section~\ref{sec:orthogonal_symmetries} can be applied to analyse the structure and dynamics of separable models, we consider three examples in detail. In particular, for the LV and SIR models we show how the orthogonal energy translation symmetries in Eq.~\eqref{eq:Gamma_u_v} act on the space of solutions. The fact that this action is transitive for all separable models, since the space of solutions is parameterised by the internal energy, means that we obtain a complete characterisation of the dynamics of the model in terms of any one solution and the corresponding generators.

Furthermore, using the ability to extend symmetries from the phase plane to the time domain we also derive analytic expressions for time domain symmetries that act by translation of the internal energy for the three examples we consider. In the case of the LV and SIR models these transformations are highly non-trivial, and, intractable to obtain by designing suitable ans\"{a}tze. In the time domain, all autonomous models have an additional time translation symmetry generated by $X=\partial_t$, which correspond to reparameterisations of a solution trajectory which preserves its internal energy. This time translation generator, together with the lifted symmetry generators $X_u$ and $X_v$ obtained from Eqs.~\eqref{eq:Y_u} and \eqref{eq:Y_v}, generate a group that acts transitively on solutions in the time domain and can be used to completely characterize the corresponding dynamics.

From a biological modelling perspective, we would finally like to emphasise the benefit of the ability, afforded by the structure of the space of solutions of separable models, to conduct a complete analysis of the separable models in terms of the internal energy and interpret the corresponding translational symmetries and their action on biologically meaningful quantities. In contrast, any symmetry of the model may be used to, e.g., compute exact solutions but the action on solutions will generally be biophysically obscure.



\section{Acknowledgements}
JGB would like to thank the Wenner--Gren Foundation for a Research Fellowship and Linacre College, Oxford, for a Junior Research Fellowship. FO would like to thank the Wolfson Centre for Mathematical Biology for hospitality and the Kempe Foundation for financial support during the conception of this work.

\section{Author contributions}
All three authors conceptualised the work, analysed the results and wrote the paper. JGB and FO formulated and proved the two theorems, constructed the examples and designed the figures.

\bibliographystyle{unsrt}
\bibliography{phase_plane_symmetry}

\begin{thebibliography}{10}

\bibitem{murray2002}
James~D. Murray.
\newblock {\em Mathematical biology. I: {A}n introduction}.
\newblock Springer--Verlag, 2002.

\bibitem{bluman1989symmetries}
George~W. Bluman and Sukeyuki Kumei.
\newblock {\em {Symmetries and differential equations}}.
\newblock Springer Science {\&} Business Media, 1989.

\bibitem{cheb2003first}
Edgardo~S. Cheb-Terrab and Theodore Kolokolnikov.
\newblock First-order ordinary differential equations, symmetries and linear
  transformations.
\newblock {\em European Journal of Applied Mathematics}, 14:231--246, 2003.

\bibitem{harrington2017reduction}
Heather~A. Harrington and Robert~A. Van~Gorder.
\newblock Reduction of dimension for nonlinear dynamical systems.
\newblock {\em Nonlinear Dynamics}, 88:715--734, 2017.

\bibitem{sen1990lie}
Tanaji Sen and Michael Tabor.
\newblock Lie symmetries of the {L}orenz model.
\newblock {\em Physica D: Nonlinear Phenomena}, 44:313--339, 1990.

\bibitem{nucci2005using}
Maria~C. Nucci.
\newblock Using {L}ie symmetries in epidemiology.
\newblock {\em Electronic Journal of Differential Equations}, 2005:87--101,
  2005.

\bibitem{ohlsson2022geometry}
Fredrik Ohlsson, Johannes~G. Borgqvist, and Ruth~E. Baker.
\newblock {On the correspondence between symmetries of two-dimensional
  autonomous dynamical systems and their phase plane realisations}.
\newblock {\em {\tt arXiv:2212.04847 [math.DS]}}, 2022.

\bibitem{hydon2000symmetry}
Peter~E. Hydon.
\newblock {\em {Symmetry methods for differential equations: {A} beginner's
  guide}}.
\newblock Cambridge University Press, 2000.

\bibitem{olver2000applications}
Peter~J. Olver.
\newblock {\em {Applications of Lie groups to differential equations}}.
\newblock Springer Science {\&} Business Media, 2000.

\bibitem{stephani1989differential}
Hans Stephani.
\newblock {\em {Differential equations: their solution using symmetries}}.
\newblock Cambridge University Press, 1989.

\bibitem{lotka1920undamped}
Alfred~J. Lotka.
\newblock Undamped oscillations derived from the law of mass action.
\newblock {\em Journal of the American Chemical Society}, 42:1595--1599, 1920.

\bibitem{lotka1925elements}
Alfred~J. Lotka.
\newblock {\em Elements of physical biology}.
\newblock Williams \& Wilkins, 1925.

\bibitem{volterra1926variations}
Vito Volterra.
\newblock Variations and fluctuations of the number of individuals in animal
  species living together.
\newblock In Royal~N. Chapman, editor, {\em Animal Ecology}, pages 409--448.
  McGraw--Hill, 1931.

\bibitem{SIR-model}
William~O. Kermack, Anderson~G. McKendrick, and Gilbert~T. Walker.
\newblock A contribution to the mathematical theory of epidemics.
\newblock {\em Proceedings of the Royal Society of London. Series A},
  115:700--721, 1927.

\bibitem{corless1996lambertw}
Robert~M. Corless, Gaston~H. Gonnet, David E.~G. Hare, David~J. Jeffrey, and
  Donald~E. Knuth.
\newblock On the {L}ambert {W} function.
\newblock {\em Advances in Computational mathematics}, 5:329--359, 1996.

\bibitem{lehtonen2016Lambert}
Jussi Lehtonen.
\newblock The {L}ambert {W} function in ecological and evolutionary models.
\newblock {\em Methods in Ecology and Evolution}, 7:1110--1118, 2016.

\bibitem{nucci2008lie}
Maria~C. Nucci.
\newblock {L}ie symmetries of a {P}ainlev{\'e}-type equation without {L}ie
  symmetries.
\newblock {\em Journal of Nonlinear Mathematical Physics}, 15(2):205--211,
  2008.

\bibitem{nucci1996complete}
Maria~C. Nucci.
\newblock The complete {K}epler group can be derived by {L}ie group analysis.
\newblock {\em Journal of Mathematical Physics}, 37(4):1772--1775, 1996.

\bibitem{nutku1990hamiltonian}
Yavuz Nutku.
\newblock Hamiltonian structure of the {L}otka-{V}olterra equations.
\newblock {\em Physics Letters A}, 145:27--28, 1990.

\end{thebibliography}
\end{document}